\documentclass[conference,twocolumn,a4paper,10pt]{IEEEtran} 
\usepackage{epsfig,cite}
\usepackage{graphicx}
\usepackage{color}
\usepackage{amssymb,amsmath}
\usepackage{mathrsfs} 
\definecolor{blue}{rgb}{0,0,1}
\definecolor{darkgreen}{rgb}{0,.5,0}
\definecolor{darkred}{rgb}{.5,0,0}
\definecolor{red}{rgb}{1,0,0}

\newcommand{\PT}[1]{\textcolor{black}{#1}}

\newtheorem{prop}{Proposition}

\def\given{\:|\:}
\def\L{\mathsf{L}}
\def\H{\mathsf{H}}
\def\U{\mathsf{U}}
\def\B{\mathsf{B}}

\newcommand{\ciid}{C_{\mathrm{IID}}}

\newcommand{\mir}{\mathcal{I}} 
\newcommand{\trient}{\mathscr{H}_3}

\def\Pr{{\mathrm{Pr}}}

\newcommand{\dfn}{\stackrel{\triangle}{=}}

\IEEEoverridecommandlockouts

\usepackage{booktabs}

\begin{document}

\title{Shannon Capacity of Signal Transduction for Multiple Independent Receptors}

\author{
	\IEEEauthorblockN{Peter J. Thomas}
	\IEEEauthorblockA{Dept. of Mathematics, Applied Mathematics, and Statistics\\ Dept. of Electrical Engineering and Computer Science\\ Dept. of Biology\\
	CWRU, Cleveland, Ohio, USA\\
	Email: pjthomas@case.edu}
	\and
	\IEEEauthorblockN{Andrew W. Eckford}
	\IEEEauthorblockA{Dept. of Electrical Engineering and Computer Science\\
	York University\\
	Toronto, Ontario, Canada\\
	Email: aeckford@yorku.ca}

\thanks{This work was in part supported by NSF grant DMS-1413770 and the Natural Sciences and Engineering Research Council (NSERC). \PT{The authors thank the five anonymous referees for their insightful comments, only some of which could be addressed in the space available.}}

}

\maketitle

\begin{abstract}
Cyclic adenosine monophosphate (cAMP) is considered a 
model system for signal transduction, the mechanism by which
cells exchange chemical messages.
Our previous work calculated the Shannon capacity of a single cAMP receptor; 
however, a typical cell may have thousands of receptors operating in parallel.
In this paper, we calculate the capacity of a cAMP signal transduction system with 
an arbitrary number of independent, indistinguishable receptors.
\PT{By leveraging prior results on feedback capacity for a single receptor, we show} (somewhat unexpectedly) that the capacity is achieved by an IID input distribution,
and that the capacity for $n$ receptors is $n$ times the capacity for a single receptor.
\end{abstract}

\section{Introduction}

In multicellular organisms, specialized cells must communicate with one another in order to coordinate
their action. The means by which a cell receives such signals is known as {\em signal transduction}:
numerous {\em receptors} on the cell's surface bind to signal-bearing
molecules, known as {\em ligands}; a bound receptor then relays the signal across the cell wall
by producing second messengers, which induce the cell to act. The ``instructions'' encoded
in signal transduction may govern tasks such as cell growth, apoptosis, differentiation, and many others.

The \textit{Dictyostelium} amoeba has on the order of 80,000 receptors uniformly distributed across its cell membrane, which are believed to act independently to transduce the cyclic adenosine monophosphate (cAMP) signal \cite{Sci:JinEtAl:2000}.  
We recently introduced a finite state channel model based on ligand-receptor binding (the BIND channel, \cite{ThomasEckford2015IEEEtrans_submission_arXiv,EckfordThomas2013ISIT}) for which we (1) rigorously obtained the capacity in a discrete time setting, (2) showed that the capacity is achieved by IID inputs, in discrete time, and (3) obtained a (non-rigorous) asymptotic expression for the mutual information rate in continuous time.  

There is a long history of research at the intersection of information theory and biology,
including notable work by Attneave \cite{Attneave1954} and Barlow \cite{Barlow1961} on 
sensory systems; Yockey \cite{Yockey1958b} on ionizing radiation and mutagenesis; 
and Berger \cite{BergerBook} on the efficiency of organ systems.
Recent progress on the computational and mathematical aspects of biology has led to a surge
of interest in biological information theory in general \cite{MitraEckford15}, and in information-theoretic analysis of signal transduction in particular (e.g., \cite{AndrewsIglesias07,thomas2003diffusion,kimmel2006information,PierobonAkyildiz11,RheeCheongLevchenko12,MahdavifarBeirami15}).
There has been much recent and related work on information-theoretic
tools with application to biological channels, such as unit output memory channels \cite{che05} (an example of which is the ``Previous Output is the STate'' (POST) channel \cite{PermuterAsnaniWeissman2014IEEETransIT}), a type of Markov channel that may be used to model receptors in signal transduction.
%
%
\PT{A recent paper \cite{TahmasbiFekri15} considers generalizations of the BIND channel to the case of  multiple receptors.} 

In \PT{the present} paper, we extend our 
previous results on the single receptor case to systems with multiple receptors,
where these receptors are {\em indistinct, independent,
and statistically identical}.  
As in earlier work, we consider discrete-time Markov models as an approximation of the receptor kinetics, modeled in continuous time
with the master equation; thus, we analyze the case where the
discrete time step $\tau \rightarrow 0$.
Our main result is to show, somewhat unexpectedly, that the capacity for $n > 1$ receptors is  achieved by an IID input distribution as $\tau \rightarrow 0$, while
the capacity for $n$ receptors is $n$ times the capacity for a single receptor.  \PT{Our analysis provides a closed form solution for the mutual information rate, valid for all $n$, that complements the asymptotic results and capacity bounds obtained in \cite{TahmasbiFekri15}.}

\section{System model}

\subsection{Model for a single receptor}

For a single receptor, we use a \PT{finite-state Markov channel model \cite{gol96}} identical to those described in 
\cite{ThomasEckford2015IEEEtrans_submission_arXiv,EckfordThomas2013ISIT}.
The cAMP receptor has two states: it may be {\em unbound}, awaiting the arrival
of a cAMP molecule; or it may be {\em bound} to cAMP, transducing the signal into
second messengers. We refer to these states as $\U$ and $\B$, respectively. (There exist far more complicated receptors, with larger state spaces; an advantage of analyzing cAMP is its 
simplicity.) 

For an individual receptor, let $p_\U(t)$ denote the probability that the receptor is in state $\U$ 
at time $t$ (resp., $p_\B(t)$ in state $\B$). It is known that this probability evolves according to a differential equation pair
(see also \cite{ThomasEckford2015IEEEtrans_submission_arXiv,Higham2008SIREV})
\begin{align}
	\label{eqn:MasterEquation}
	\frac{d p_\U(t)}{dt} &= - k_+ c(t) p_\U(t) + k_- p_\B(t) \\
	\label{eqn:MasterEquation2}
	\frac{d p_\B(t)}{dt} &= k_+ c(t) p_\U(t) - k_- p_\B(t)  
\end{align}
where $c(t)$ is the concentration of cAMP, and $k_+$ and $k_-$ are rate constants, corresponding to the $\U\rightarrow\B$ and
$\U\leftarrow\B$ reactions, respectively. 
Following the {\em principle of mass action}, $\U\rightarrow\B$ requires a cAMP molecule, therefore its rate is proportional to $c(t)$; however, $\U\leftarrow\B$ requires no external molecules, and its rate is therefore independent of $c(t)$.

The model in (\ref{eqn:MasterEquation})-(\ref{eqn:MasterEquation2}) 
can be \PT{approximated} by a discrete-time Markov chain,
and we take advantage of this discretization in obtaining our results.
We assume that the concentration $c(t)$ is binary:
$c(t) \in \{\L,\H\}$, where $\L$ is the lowest possible concentration, and $\H$ is
the highest possible concentration.%
\footnote{In \cite{ThomasEckford2015IEEEtrans_submission_arXiv}
we show that binary inputs achieve capacity for a single receptor \PT{(discrete time case)}. We expect this result to hold
for an arbitrary number of receptors, but in this paper we \PT{restrict attention to binary input distributions for simplicity.}
}
Let $\alpha_\H = k_+ \H$, $\alpha_\L = k_+ \L$, and $\beta = k_-$; further, 
let $\tau$ represent a discrete time step. The discrete-time approximation 
for the differential equation pair (\ref{eqn:MasterEquation})-(\ref{eqn:MasterEquation2})
is given by
\begin{align}
	p_\U(t + \tau) &= (1-\tau \alpha_{\H/\L}) p_\U(t) + \tau \beta p_\B(t) \PT{+o(\tau)}\\
	p_\B(t + \tau) &= \tau \alpha_{\H/\L} p_\U(t) + (1-\tau \beta)p_\B(t) \PT{+o(\tau)}
\end{align}
where $\alpha_{\H/\L}$ should be replaced with either $\alpha_\H$ or $\alpha_\L$, depending on
the concentration\PT{, and $f=o(\tau)$ means $\lim_{\tau\to 0}(f(\tau)/\tau=0)$.}
\PT{Neglecting terms $o(\tau)$, the channel state can be represented}
as a discrete-time Markov chain, with transition probability matrix
\begin{equation}
	\mathbf{P}_{\H/\L} = \left[ \begin{array}{cc} 1-\tau \,\alpha_{\H/\L} & \tau\, \alpha_{\H/\L} \\ \tau\, \beta & 1-\tau \,\beta \end{array} \right] .
\end{equation}
%
%
The state transition diagram for cAMP is given in Figure \ref{fig:single}.
\begin{figure}[t!]
	\begin{center}
	\includegraphics[width=2in]{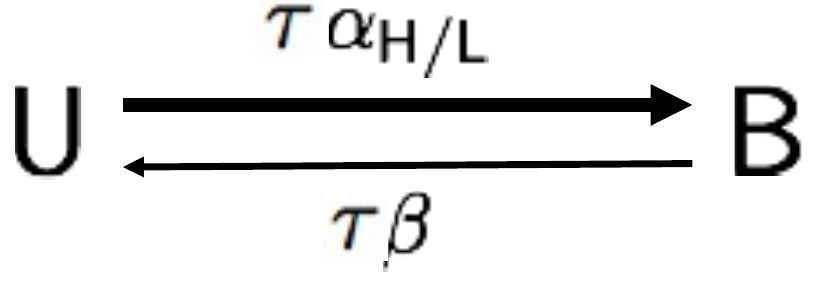}
	\end{center}
	\caption{\label{fig:single} State transition diagram for a \PT{single} cAMP receptor \PT{with discrete time step $\tau$}; arrows are marked with their transition probabilities. \PT{The input-sensitive state transition is indicated} by a bold arrow. See also \cite{EckfordThomas15}.}
\end{figure}


\subsection{Multiple receptors}

Suppose we have $n$ identical, independent receptors each with \PT{individual} binding probabilities \PT{$\tau\alpha_\L$, 
$\tau\alpha_\H$, and $\tau\beta$,} as defined above. 
We consider therefore a model with $n+1$ distinct states representing a population of $n$ indistinguishable receptors: state $k$ refers to the system with $k$ out of $n$ receptors bound to signaling (ligand) molecules.  If each receptor binds or unbinds signaling molecules \emph{independently} of the other receptors, then the state transition probabilities are as given in Figure \ref{fig:TransitionDiagram}.
\begin{figure*}[t!]
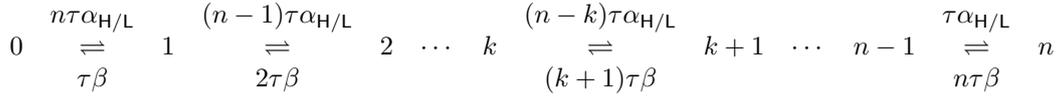

\begin{equation*}
\begin{array}{ccccccccccccc}
&n\tau\alpha_{\H/\L} &&(n-1)\tau\alpha_{\H/\L} &&&&(n-k) \tau\alpha_{\H/\L}&&&&\tau\alpha_{\H/\L}&\\
0&\rightleftharpoons&1&\rightleftharpoons&2&\cdots&k&\rightleftharpoons&k+1 & \cdots & n-1 & \rightleftharpoons & n\\
&\tau\beta &&2\tau\beta &&&&(k+1)\tau\beta&&&&n\tau\beta&
\end{array}
\end{equation*}
\caption{\label{fig:TransitionDiagram} State transition diagram for $n$ independent\PT{, indistinguishable} cAMP receptors obeying mass-action kinetics.}
\end{figure*}
%
From the figure,
because all $n$ receptors could potentially be bound (or unbound) at any one time, we require that the individual receptor binding or unbinding probabilities be sufficiently small, i.e.~\PT{$\text{max}(\tau\alpha_\H,\tau\beta)< \frac1n$.  For a given set of reaction rates, this condition can be met by reducing the size of $\tau$, the discrete time increment}.  Therefore for systems with large numbers of receptors, the continuous time setting is ultimately more natural than discrete time; moreover molecular communication systems do not generally have access to a reference clock as is normally the case in macroscopic engineered systems.   Following our previous work, however, we begin the analysis assuming a (small) discrete time step, and later consider the $\tau\to0$ limit of our mutual information results. 

 \PT{Channel definition: The input is a sequence of ligand concentrations $X_i\in\{\L,\H\}$. The output (also the state) is the number of bound receptors $Y_i\in\{0,1,\ldots,n\}$. The state transitions obey $\left(\mathbf{P}_{\H/\L}\right)_{k,k+1}=(n-k)\tau\alpha_{\H/\L}, \,0\le k \le n-1,$ and $\left(\mathbf{P}_{\H/\L}\right)_{k,k-1}=k\tau\beta,\, 1\le k \le n$; see the next section for more details. We require $\tau \,n\, \text{max}(\alpha_\H,\beta)<1$.}
This channel, which we call the $\text{BIND}^n(\alpha_\L,\alpha_\H,\beta)$ channel, is a member of the set of Chen-Berger unit output memory channels \cite{che05}. For all such channels,  
the feedback-capacity-achieving
input distribution has the form
\begin{equation}
	\PT{p(X^m || Y^m) = \prod_{i=1}^m} p_{X_i | Y_{i-1}}(x_i \given y_{i-1}) ,
\end{equation}
where $||$ represents causal conditioning; further,
$y_0$ is null. That is, each input is dependent {\em only on the previous state of the channel},
and no other past inputs or states.

Now consider an encoding scheme, exploiting feedback, in which the probability of sending input $\H$ when the channel is in state $k$ is $p_k := p_{X_i | Y_{i-1}}(\H \given k)$, for $0\le k \le n$.  When the channel is in the fully bound state ($Y=n$) each individual receptor releases its ligand with probability $\beta$, independent of the input concentration, so the choice of $p_n$ has no effect on information transmission (see also \cite{EckfordThomas2013ISIT}).  The capacity therefore requires optimizing over the \PT{$n$ free parameters} $p_0,\ldots,p_{n-1}$. 

\section{Results}

\subsection{Mutual information for two receptors}

First we consider the case $n=2$ in detail, and then generalize to arbitrary $n$.  
 With two receptors we have
\begin{equation}
	\begin{array}{ccccc}
		&2\tau\alpha_{\H/\L} &&\tau\alpha_{\H/\L} &\\
		0&\rightleftharpoons&1&\rightleftharpoons&2\\
		&\tau\beta &&2\tau\beta &
	\end{array}
\end{equation}
Thus we have a transition probability matrix
\begin{equation}\label{eq:transitionprobmatrix}
	\mathbf{P}_{\H/\L} = 
	\left[
		\begin{array}{ccc}
			1 - 2 \tau\alpha_{\H/\L} & 2 \tau\alpha_{\H/\L} & 0 \\
			\tau\beta & 1 - \tau(\beta + \alpha_{\H/\L}) & \tau\alpha_{\H/\L} \\
			0 & 2 \tau\beta & 1 - 2\tau\beta
		\end{array}
	\right]
\end{equation}
Using the capacity-achieving input distribution, the output states $Y^n$ form a Markov
chain (letting $\alpha_k = \alpha_\L + (\alpha_\H - \alpha_\L)p_k$):
\begin{equation}
	\mathbf{P}_Y = 
	\left[
		\begin{array}{ccc}
			1 - 2 \tau\alpha_0 & 2 \tau\alpha_0 & 0 \\
			\tau\beta & 1 - \tau(\beta + \alpha_1) & \tau\alpha_1 \\
			0 & 2 \tau\beta & 1 - 2\tau\beta
		\end{array}
	\right]
\end{equation}
Thus, with binary inputs $\H,\L$, the state is
completely described by two parameters, represented by $p_k$ for $k \in \{0,1\}$.
(For $k=2$, the fully bound case, there is no information transmission.)

The steady-state distribution of $Y$ is given by the Perron-Frobenius eigenvector.
Letting $\pi_i = \Pr(Y = i)$:
\begin{align}
	\pi_0 &:= \Pr(Y = 0) = \frac{\beta^2}{\beta^2 + 2 \alpha_0 \beta + \alpha_0 \alpha_1} \\
	\pi_1 &:= \Pr(Y = 1) = \frac{2 \alpha_0 \beta}{\beta^2 + 2 \alpha_0 \beta + \alpha_0 \alpha_1} \\
	\pi_2 &:= \Pr(Y = 2) = \frac{\alpha_0 \alpha_1}{\beta^2 + 2 \alpha_0 \beta + \alpha_0 \alpha_1} .
\end{align}
The mutual information rate is given by
\begin{align}
	\mir(X;Y) &= \PT{\lim_{m \rightarrow \infty}\frac{1}{m} I(X_1^m,Y_1^m)}\\ \nonumber
	&= H(Y_i \given Y_{i-1}) - H(Y_i \given X_i, Y_{i-1}),\text{ \PT{for arbitrary} }i.
\end{align}
Let $\phi(p)$ represent the {\em partial entropy function}, where
\begin{equation}
	\phi(p) = 
	\left\{
		\begin{array}{cl}
			0, & p = 0\\
			- p \log p, & p \neq 0
		\end{array}
	\right.
\end{equation}
(we use natural logarithms throughout, so information is measured in nats).
Further let $\trient(p,q)$ represent the {\em triple entropy function}, where
\begin{equation}
	\trient(p,q) = \phi(p) + \phi(q) + \phi(1-p-q),
\end{equation}
defined for $p+q\le 1$.
(Note that $\trient(p,0)$ reduces to the binary entropy function.) Then
\begin{align}
	H(Y_i \given Y_{i-1} = 0) &= \trient(2 \tau\alpha_0,0) \\
	H(Y_i \given Y_{i-1} = 1) &= \trient(\tau\alpha_1,\tau\beta) \\
	\label{eqn:insensitive-1}
	H(Y_i \given Y_{i-1} = 2) &= \trient(2 \tau\beta,0)
\end{align}
and
\begin{align*}
	H(Y_i | X_i, Y_{i-1} = 0) &= p_0 \trient(2 \tau\alpha_\H,0) + (1-p_0) \trient(2 \tau\alpha_\L,0)\\
	H(Y_i | X_i, Y_{i-1} = 1) &= p_1 \trient(\tau\alpha_\H,\tau\beta) + (1-p_1) \trient(\tau\alpha_\L,\tau\beta)\\
	H(Y_i | X_i, Y_{i-1} = 2) &= p_2 \trient(2 \tau\beta,0) + (1-p_2) \trient(2 \tau\beta,0)
\end{align*}
In 
\PT{the preceding equation,} we can reduce to $H(Y_i \given X_i, Y_{i-1} = 2) = \trient(2 \tau\beta,0)$,
the same as $H(Y_i \given Y_{i-1} = 2)$ in (\ref{eqn:insensitive-1}). Finally,
\begin{align}
	\label{eqn:MutualInformationRate} 
	\lefteqn{\mir(X;Y) = }& \\ 
	\nonumber 
	& \pi_0 \Big( \trient(2 \tau\alpha_0,0) - p_0 \trient(2 \tau\alpha_\H,0) - (1-p_0) \trient(2 \tau\alpha_\L,0) \Big) \\
	\nonumber 
	& + \pi_1 \Big( \trient(\tau\alpha_1,\tau\beta) - p_1 \trient(\tau\alpha_\H,\tau\beta) - (1-p_1) \trient(\tau\alpha_\L,\tau\beta) \Big).
\end{align}

\subsection{As $\tau \rightarrow 0$, capacity-achieving input distribution is IID}

In order for the capacity-achieving distribution to be IID, it must be true that 
(\ref{eqn:MutualInformationRate}) is maximized with $p_0 = p_1$ (which
implies $\alpha_0 = \alpha_1$).
It can be shown via numerical examples that the capacity-achieving input distribution is
{\em not} IID for {\em arbitrary} values of $\alpha_\L$, $\alpha_\H$, and $\beta$, \PT{and finite $\tau>0$}. 
However, we are interested in the \PT{limiting case where $\tau \rightarrow 0$.}


The reader may check that the state occupancy probabilities $\pi_i$ are independent of $\tau$. However, 
(\ref{eqn:MutualInformationRate}) becomes
\begin{align}
	\nonumber
	\lefteqn{\mir(X;Y) = }& \\ 
	\nonumber 
	& \frac{\pi_0}{\tau} \Big( \trient(2 \tau \alpha_0,0) \\ 
	\nonumber
	&- p_0 \trient(2 \tau \alpha_\H,0) - (1-p_0) \trient(2 \tau \alpha_\L,0) \Big) \\
	\nonumber
	& + \frac{\pi_1}{\tau} \Big( \trient(\tau \alpha_1,\tau\beta) \\ 
	\label{eqn:MutualInformationRate-tau} 
	&- p_1 \trient(\tau \alpha_\H,\tau\beta) - (1-p_1) \trient(\tau \alpha_\L,\tau\beta) \Big).
\end{align}
\PT{Although terms such as $\trient(2\tau\alpha_0,0)$ diverge as $\tau\to 0$, the divergent terms cancel in \eqref{eqn:MutualInformationRate-tau}, and $\mir(X;Y)$ remains finite in the limit.  Straightforward application of l'Hopital's rule yields}
\begin{align}
	\nonumber \lim_{\tau \rightarrow 0} & \frac{\pi_0}{\tau} \Big( \trient(2 \tau \alpha_0,0) \\ \nonumber & - p_0 \trient(2 \tau \alpha_\H,0) - (1-p_0) \trient(2 \tau \alpha_\L,0) \Big) & \\
	&= 2 \pi_0 \Big( \phi(\alpha_0) - p_0 \phi(\alpha_\H) - (1-p_0) \phi(\alpha_\L) \Big) \\ 
	\nonumber \lim_{\tau \rightarrow 0} & \frac{\pi_1}{\tau} \Big( \trient(\tau \alpha_1,\tau\beta) \\ \nonumber &- p_1 \trient(\tau \alpha_\H,\tau\beta) - (1-p_1) \trient(\tau \alpha_\H,\tau\beta) \Big) & \\
	&= \pi_1 \Big( \phi(\alpha_1) - p_1 \phi(\alpha_\H) - (1-p_1) \phi(\alpha_\L) \Big)
\end{align}
Finally, letting $\mathcal{Z}=\beta^2+2\alpha_0\beta+\alpha_0\alpha_1$,
\begin{align}
	\nonumber
	\lim_{\tau \rightarrow 0} \mir(X;Y)
\nonumber
=& \frac{2\beta}{\mathcal{Z}}\Big(
\beta\phi(\alpha_0)+\alpha_0\phi(\alpha_1)
\\ \nonumber &-(\beta p_0+\alpha_0 p_1)\phi(\alpha_\H)
\\&-((\alpha_0+\beta)-\beta p_0-\alpha_0 p_1)\phi(\alpha_\L)
  \Big)
\end{align}

To proceed, we assume that the input distribution is IID, and show that this is optimal.
Under the IID constraint we have $p_0=p_1=p$, which leads to several simplifications.  The stationary probabilities become  binomial: writing $\bar{\alpha}=\alpha_0=\alpha_1=\alpha_2$,
\begin{equation}
\pi_k(p)={2\choose k}\frac{\bar{\alpha}^k\PT{\beta^{2-k}}}{(\bar{\alpha}+\beta)^2},\text{ and }\mathcal{Z}=(\bar{\alpha}+\beta)^2,
\end{equation}
and the continuous time information rate reduces to
\begin{align}
	\nonumber
	\lefteqn{\lim_{\tau \rightarrow 0} \mir(X;Y)|_{p_k\equiv p} }& \\
	\label{eqn:2ReceptorMutualInformationRate}
	&= 
	2\left(\frac{\beta}{\bar{\alpha}+\beta}\right)\left( \phi(\bar{\alpha})-p\phi(\alpha_\H)-(1-p)\phi(\alpha_\L)  \right).
\end{align}
The IID capacity of this channel is given by
\begin{align}
	\ciid &= \max_{p} \lim_{\tau \rightarrow 0} \mir(X;Y)|_{p_k\equiv p}\\
	\label{eqn:2ReceptorCapacity}
	&= 2 \max_{p} \left(\frac{\beta}{\bar{\alpha}+\beta}\right)\left( \phi(\bar{\alpha})-p\phi(\alpha_\H)-(1-p)\phi(\alpha_\L)  \right) .
\end{align}

Thus we may state:
\begin{prop}
 For two receptors, as $\tau \rightarrow 0$, the capacity-achieving input distribution is IID.
\end{prop}
\begin{proof}
Let $C$ represent the capacity.
If we take $\ciid/2$ in (\ref{eqn:2ReceptorCapacity}), we obtain exactly the capacity for a single receptor as $\tau \rightarrow 0$, from \cite{EckfordThomas2013ISIT,ThomasEckford2015IEEEtrans_submission_arXiv}.

We know that $C \geq \ciid$. 
The capacity two independent receptors can be no greater than twice the capacity of a single
receptor, so $C \leq 2 (\ciid/2) = \ciid$. Since $C$ is bounded above and below by $\ciid$,
the result follows.
\end{proof}

We give a numerical example of this result in Figure \ref{fig:multiple-receptors}, where the global maximum of the MI rate indeed lies on the diagonal $p_0=p_1=p$.  

\begin{figure}[htbp] 
   \centering
   \includegraphics[width=2.5in]{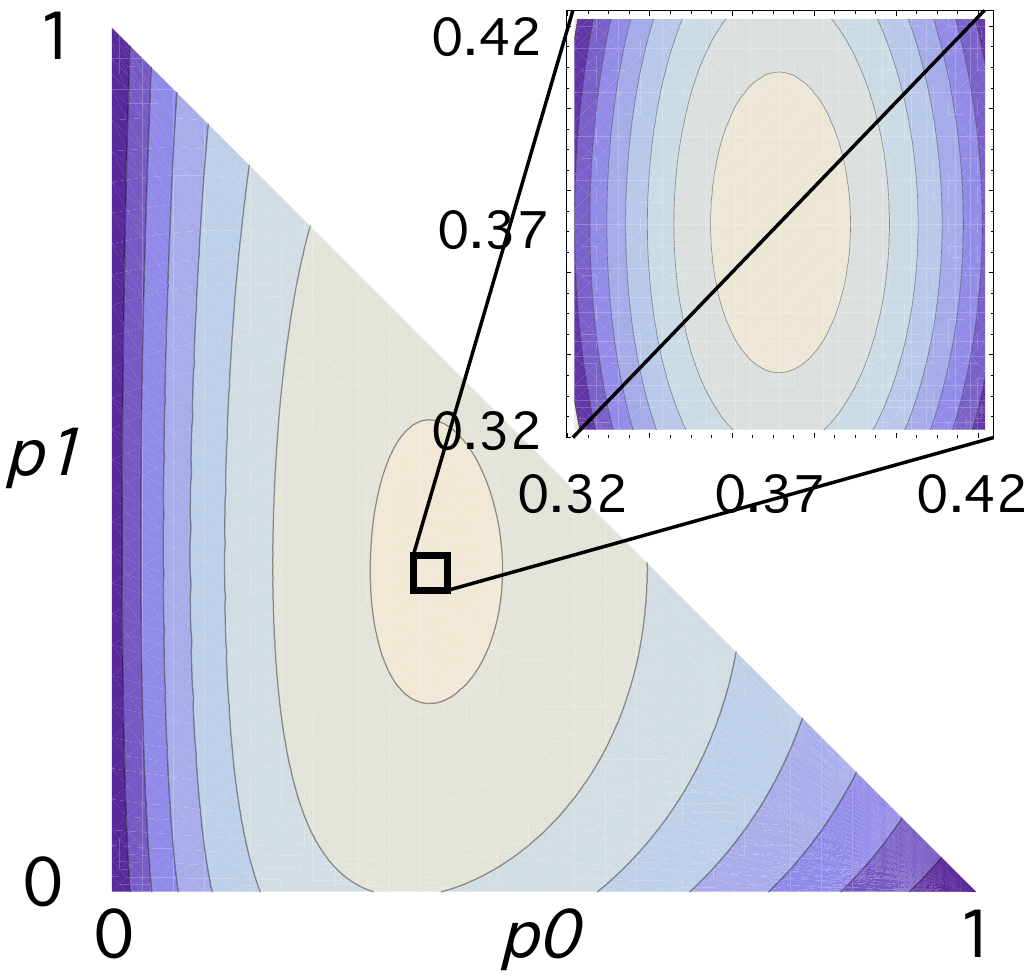} 
   \caption{Mutual information rate for two independent receptors, with feedback, as a function of the high-input probability when unbound ($p_0$) and when singly bound ($p_1$).  The high-input probability when the receptor is doubly bound ($p_2$) does not affect the MI rate.  For this example the continuous time transition rates are  $\alpha_\L=1, \alpha_\H=10, \beta=20$.  All rates are in Hz. The optimal MI rate of 3.57367 nats per second is attained when $p_0=p_1\approx \PT{0.371696}$. Diagonal  inside the zoomed-in box marks the line $p_0=p_1$.}
   \label{fig:multiple-receptors}
\end{figure}

In \cite{EckfordThomas2013ISIT,ThomasEckford2015IEEEtrans_submission_arXiv} we showed rigorously that the feedback capacity and the IID capacity for the single receptor were identical (this result held for all values of \PT{$\tau>0$}, i.e., all settings of the parameters).  On the other hand, the feedback capacity for the two receptors cannot exceed twice the feedback capacity of a single receptor, when the receptors are independent, because we could consider the two independent receptors separately.  Therefore the system with two identical, independent receptors inherits the property that feedback capacity = IID capacity = capacity.  

\subsection{Generalizing to $n > 2$}

Returning to arbitrary $n > 2$, we adopt the following notation.  Let $\bar{\alpha}_k(p_k)=\alpha_\L+p_k(\alpha_\H-\alpha_\L)$ be the average \textit{per capita} transition rate from state $k$ to state $k+1$.  
Define
\begin{align}
A_k&=\left\{
\begin{array}{ll}
\beta^n,&k=0\\
{n\choose k}\beta^{n-k}\prod_{j=0}^{k-1}\bar{\alpha}_j(p_j),&  1\le k \le n
\end{array}\right.\\
\mathcal{Z}(p_0,\ldots,p_{n-1})&=\sum_{k=0}^n A_k.
\end{align}
Then $\pi_k \dfn \Pr\{Y=k\}=A_k/\mathcal{Z}$ is the stationary probability of the $k$-bound state.  If we take $p_k=p$ for all $k$ (the case of IID inputs) then we have the same \textit{per capita} binding rate in each state, $\bar{\alpha}_k=\bar{\alpha}=\alpha_\L+p(\alpha_\H-\alpha_\L)$.  In this case the stationary distribution is binomial:
\begin{equation}
\pi_k^{\text{IID}}={n\choose k}\left( \frac{\bar{\alpha}}{\bar{\alpha}+\beta} \right)^k\left( \frac{\beta}{\bar{\alpha}+\beta} \right)^{n-k}
\end{equation}
where $\bar{\alpha}/(\bar{\alpha}+\beta)$ is the equilibrium probability of any given receptor being in the bound state.  

Fixing some $n>1$, we write $I(\mathbf{p})$, the average mutual information rate per time step, for $I(\alpha_\L,\alpha_\H,\beta,p_0,\ldots,p_{n-1})$, i.e.~we leave the dependence on the transition probabilities implicit. The vector $\mathbf{p}=[p_0,p_1,\ldots,p_{n-1}]^\intercal$ represents the high (\textit{versus} low) input probabilities for the feedback case.  The IID case corresponds to $\mathbf{p}=\mathbf{1}p\equiv[p,p,\ldots,p]^\intercal$. 
$I(\mathbf{p})$ may be written as a sum over the information rates contributed by each edge: $I=\sum_{k=0}^{n-1}I_k=\sum_{k=0}^{n-1}\pi_k \Phi_k$, 
%
where
\begin{align}
\nonumber
\Phi_k =& \phi\left((n-k)\bar{\alpha}_k\right)\\&-\Big((1-p_k)\phi((n-k)\alpha_\L)+p_k\phi((n-k)\alpha_\H)    \Big)
\\
=& (n-k)\Psi(p_k)
\end{align}
with 
$\Psi(p)=\phi\Big(p\alpha_\H+(1-p)\alpha_\L\Big)-\Big(p\phi(\alpha_\H)+(1-p)\phi(\alpha_\L)\Big).$ 
Here we have used the product rule property of the partial entropy, $\phi(kp)=k\phi(p)+p\phi(k)$.

As in the case of $n=2$ receptors, the IID case simplifies to
\begin{align}\nonumber
I(\mathbf{1}p)&=\sum_{k=0}^{n-1}\pi_k(n-k)\Psi(p)
=\Psi(p)\sum_{k=0}^{n-1}{n\choose k}\frac{\bar{\alpha}^k\beta^{n-k}}{(\bar{\alpha}+\beta)^n}(n-k)\\
&=\Psi(p)\left(n-n\frac{\bar{\alpha}}{\bar{\alpha}+\beta}\right) = n\Psi(p)\left(\frac{\beta}{\bar{\alpha}+\beta}\right)
\end{align}
which is $n$ times the mutual information rate for a single receptor. This leads to the following.
\begin{prop}
	For $n$ receptors, as $\tau \rightarrow 0$, capacity $C$ is $n$ times the single-receptor capacity,
	and the capacity-achieving input distribution is IID.
\end{prop}
\begin{proof}
By a similar argument to Proposition 1, 
considering $n$ independent distinguishable receptors receiving the same input signal, it is clear that the $n$-receptor information rate with feedback cannot exceed $n$ times the single receptor rate with feedback.  Therefore for $n$ indistinguishable receptors, $C_\text{IID}=C$,  for arbitrary $n$, and the result follows.
\end{proof}

In contrast, \emph{non-independent} receptors can have feedback capacity greater than IID capacity.  \PT{Fig.~\ref{fig:cooperative} shows this effect for a channel for which ligand binding is \emph{cooperative} instead of independent; e.g.~with transition probability matrix (\textit{cf}.~\eqref{eq:transitionprobmatrix}) 
\begin{equation}
	\mathbf{P}_{\H/\L} = 
	\left[
		\begin{array}{ccc}
			1 -  \tau\alpha_{\H/\L} &  \tau\alpha_{\H/\L} & 0 \\
			\tau\beta & 1 - \tau(\beta + \alpha_{\H/\L}) & \tau\alpha_{\H/\L} \\
			0 &  \tau\beta & 1 - \tau\beta
		\end{array}
	\right].
\end{equation}}

\begin{figure}[htbp] 
   \centering
   \includegraphics[width=2.5in]{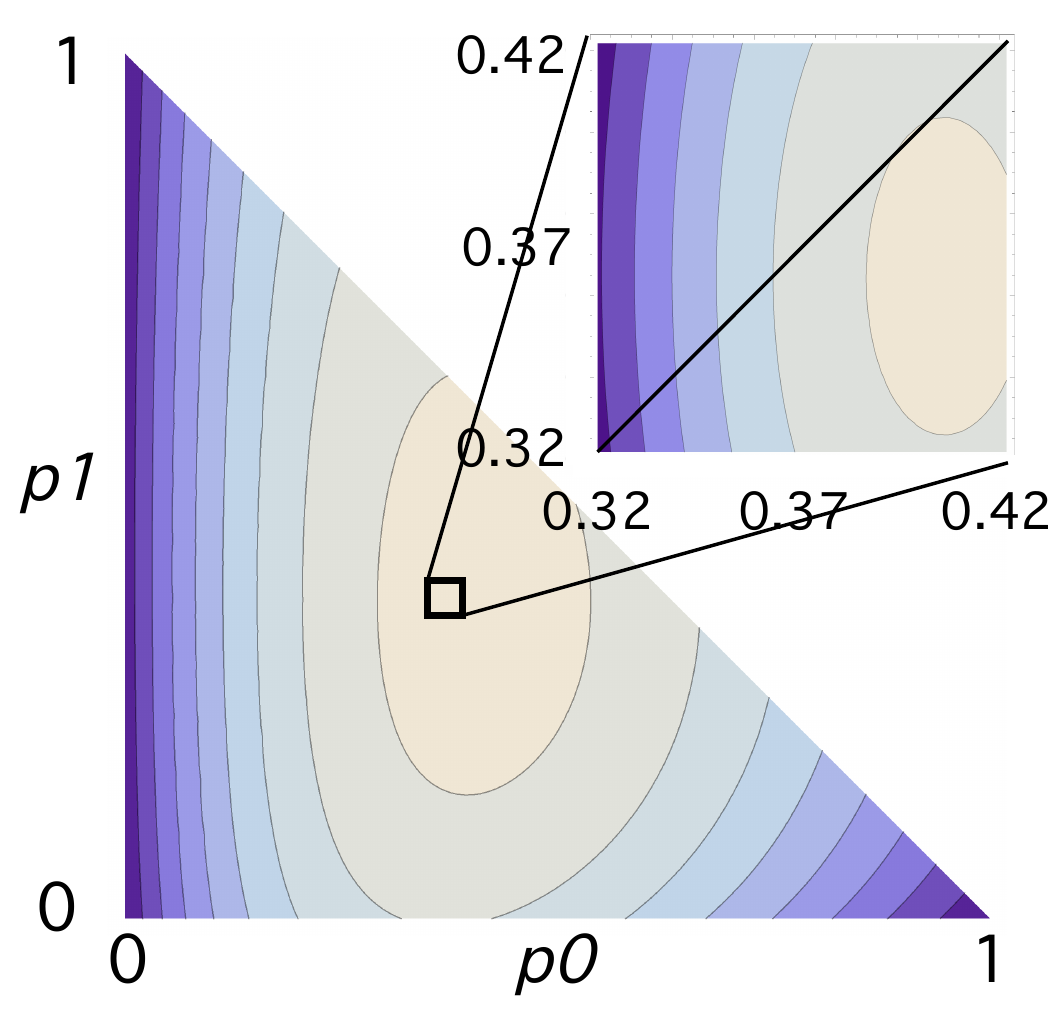} 
   \caption{Mutual information rate, with feedback, for two receptors, as a function $p_0$ and $p_1$, when the receptors are \emph{not independent}.  Total transition rates: $i\to i+1$, $\alpha_{i,\H}=10$Hz or $\alpha_{i,\L}=1$Hz; $i\to i-1$, $\beta_i=20$Hz.  The optimal MI rate of 2.1026 nats per second is attained when $p_0\approx 0.407$, $p_1\approx 0.364$. Zoom-in box and $p_0=p_1$ line positions same as Fig.~\ref{fig:multiple-receptors}.}
   \label{fig:cooperative}
\end{figure}

\section{Discussion}
Analysis of the single-receptor BIND channel introduced in \cite{ThomasEckford2015IEEEtrans_submission_arXiv,EckfordThomas2013ISIT} turns on the observation that it falls in the class of channels satisfying the Chen-Berger condition \cite{BergerYing2003ieeeISIT,che05}.  At the same time, it may be viewed as an instance of a POST channel (Previous Output is the STate) \cite{PermuterAsnaniWeissman2014IEEETransIT}.  Here we show that a ligand-binding system comprising $n$ identical, independent receptors also satisfies these conditions.  The $n$-receptor BIND channel has a mutual information rate equal to $n$ times the mutual information rate of the single receptor BIND channel, and its capacity is realized by an IID input source with optimal high concentration probability $p_*$ that is the same for one receptor as it is for $n$ receptors.  Under the IID input condition the receptors' stationary state distribution becomes binomial. suggesting a relation to the well-studied binomial channel  
\cite{KomninakisVandenbergheWesel2001IEEE_ISIT}.
\PT{As discussed in \cite{ThomasEckford2015IEEEtrans_submission_arXiv}, the physical channel model breaks down as $\tau\to 0$ in the sense that the concentration at the receiver will not remain IID at arbitrarily fine time scales.  The way in which biophysical constraints restrict the input ensemble will be system specific, and a topic for future investigations.}


\bibliographystyle{IEEEtran}

\end{document}